\documentclass[reprint,showpacs,preprintnumbers,amsmath,amssymb,aps,pra]{revtex4-1}
\usepackage{graphicx}
\usepackage{dcolumn}
\usepackage{amsthm,amsfonts}
\usepackage[dvipdfmx]{hyperref}

\begin{document}

\title{Localization in  Quantum Walks on a Honeycomb Network}
\author{Changyuan Lyu}
\author{Luyan Yu}
\author{Shengjun Wu}
\affiliation{Kuang Yaming Honors School, Nanjing University, Nanjing, Jiangsu 210023, China}
\date{\today}

\begin{abstract}
We systematically study the localization effect in discrete-time quantum walks on a honeycomb network and establish the mathematical framework. We focus on the Grover walk first and rigorously derive the limit form of the walker's state, showing it has a certain probability to be localized at the starting position. The relationship between localization and the initial coin state is concisely represented by a linear map. We also define and calculate the average probability of localization by generating random initial states. Further, coin operators varying with positions are considered and the sufficient condition for localization is discussed. We also similarly analyze another four-state Grover walk. Theoretical predictions are all in accord with numerical simulation results. Finally, our results are compared with previous works to demonstrate the unusual trapping effect of quantum walks on a honeycomb network, as well as the advantages of our method.
\end{abstract}

\maketitle

\newcommand{\ket}[1]{|#1\rangle}
\newcommand{\ketp}[1]{\ket{#1}_p}
\newcommand{\bra}[1]{\langle #1|}
\newcommand{\inp}[2]{\langle #1|#2\rangle}
\newcommand{\dif}{\mathrm{d}}
\newtheorem{theorem}{\textbf{Theorem}}
\newtheorem{lemma}{\textbf{Lemma}}
\newtheorem{proposition}{\textbf{Proposition}}

\section{Introduction}
Quantum walks are generally classified into two categories: discrete-time and continuous-time quantum walks~\cite{PhysRevA.48.1687,Meyer1996,ComprehensiveReview,PhysRevA.74.030301}. In this paper we focus on the former which were first introduced by Aharonov \textsl{et al.} in 1993. A discrete-time quantum walk usually involves a quantum coin and a walker. The coin controls the walker's movement. Since a quantum coin could be in any superposition state of up and down, quantum walks exhibit many different features when compared to classical random walks (see Refs.~\cite{Mackay2002,childs2002,PhysRevE.69.026119,Carneiro2005,Schmitz09,Konno2002Quantum}). For example, the quantum walker spreads at a rate linear in time in the one-dimensional Hadamard walk~\cite{Konno2002Quantum}, much faster than the square-root rate in the classical counterpart.\par
Another uncommon phenomenon in quantum walks is localization. For a classical random walk on a line or a square network, as well as the one-dimensional Hadamard walk, the probability of finding the walker at a specific position approaches zero when the number of steps is large enough~\cite{ComprehensiveReview,PhysRevE.72.056112}. However, Tregenna  \textsl{et al.} numerically discovered that in quantum walks on a 2D square lattice, the walker under the control of Grover's operator~\cite{Grover1996} was highly probable to be found at its initial position, which was called the walker's localization~\cite{Tregenna2003}. Inui \textsl{et al.} then rigorously analyzed this phenomenon and owed localization to degenerate eigenvalues of the system's time-evolution operator~\cite{PhysRevA.69.052323}. Later localization was also found in quantum walks on other kinds of lattices~\cite{PhysRevE.72.056112} or with other coin operators~\cite{PhysRevA.77.062331,PhysRevLett.100.020501,PhysRevA.78.032306,PhysRevA.91.022308}. Additional time-dependent phase gates could also cause localization in a 2D alternate Hadamard walk according to literature~\cite{PhysRevA.91.012328}.\par
Some quantum walks have been successfully employed to design effective quantum search algorithms~\cite{AMBAINIS2003,PhysRevA.67.052307,Abal2010,Childs:2003db,Childs2003}, whose target is to transform the system to a specific state at a probability of $O(1)$. Thus the localization effect in quantum walks shows potential applications in designing fast search algorithms. Some simple quantum walk models have already been implemented by experiments~\cite{PhysRevA.61.013410,PhysRevA.65.032310,PhysRevLett.100.170506,Karski10072009,PhysRevA.86.052327,Jeong:2013cg,wang2013physical}. As the structure of graphene, the honeycomb lattice has been paid much attention to recently. It may have application prospects in quantum controlling and computing based on 2D hexagonal materials like the graphene and h-BN~\cite{Pakdel2012256}. Besides, since the honeycomb lattice is more complicated than a line or a common 2D square lattice, we could expect more different characteristics in quantum walks on a honeycomb network. \par
This paper is organized as follows: Sec.\ref{QWframework} is devoted to the definition and mathematical framework of quantum walks on a honeycomb network. %We first focus on Grover walk and obtain the analytical expression of the state vector.
In Sec.~\ref{Limit of the state ket}, we calculate the ket in the limit $t\rightarrow\infty$ and show the existence of localization in the Grover walk, which is in accord with numerical simulations. Then in Sec.~\ref{Conditions and Dependence of Localization}, we investigate the dependence of localization on the initial state and derive a non-zero average localization probability. Quantum walks governed by other coin operators are considered and the condition for localization is proposed. Furthermore in Sec.~\ref{4DCondition}, an extended model with a four-state coin is analyzed. In Sec.~\ref{Comparison}, we compare our results to previous models and show the different properties of quantum walks on a honeycomb network.
\section{Quantum Walks on a honeycomb network}\label{QWframework}
\subsection{\label{twoSpaces}The position space $\mathcal{H}_p$ and the coin space $\mathcal{H}_c$}
The quantum walk of a particle is usually described by two Hilbert spaces, the coin space $\mathcal{H}_c$ and the position space $\mathcal{H}_p$. Different from an infinite square lattice or an unrestricted line, a honeycomb network is not a kind of Bravais lattice, for here are two types of vertices~\cite{Kittel,Zhai2014}. For example, after we translate the network along the vector $\vec{E}_1$ in FIG.~\ref{honeycomb}, the new lattice is not identical to the original one. Besides, orthogonal Cartesian coordinates, e.g. $X$-$Y$ in FIG.~\ref{honeycomb}, are inconvenient for describing the positions of vertices here. Therefore, we divide the network into rhombuses and set up an oblique coordinate system. A vertex is labeled as $(x,y,s)$, where $x$ and $y$ stand for coordinates of the lower left corner of the rhombus in which the vertex is contained, and $s$ indicates the type of the vertex. We make the convention that a black (white) vertex is type-0 (type-1), thus any position eigenket is denoted as $\ketp{x,y,s}$, or $\ketp{\vec{r}, s}$ in short, where $\vec{r}$ represents $(x,y)$.\par
\begin{figure}[t]
\centering
\includegraphics[width=0.4\textwidth]{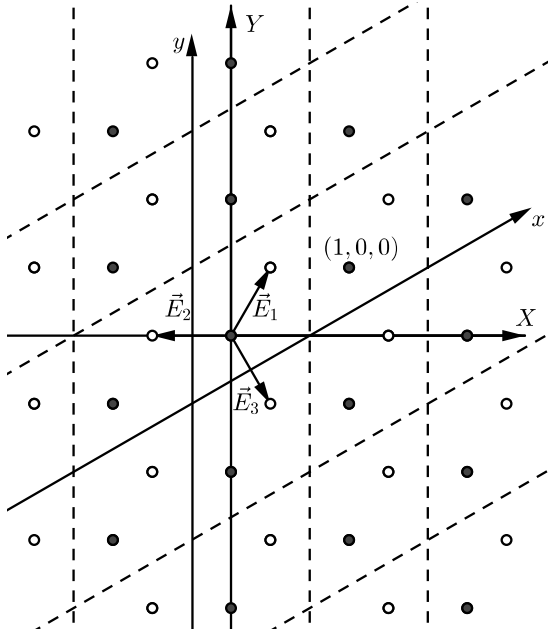}
\caption{Honeycomb network. Black and white points are two types of vertices. The distance between two nearest vertices is set to $\frac1{\sqrt 3}$.} \label{honeycomb}
\end{figure}
On a honeycomb lattice, a particle is able to move in three directions, so the dimension of the coin space is three. The eigenstates of the coin are labeled as $\ket{1}$, $\ket{2}$, and $\ket{3}$, corresponding to movement along (against) $\vec{E}_1$, $\vec{E}_2$ and $\vec{E}_3$ from a type-0 (type-1) vertex. Whether the walker will move \emph{along} or \emph{against} these vectors is determined by the \emph{type} of the vertex. The system's total state vector is denoted by a ket in the composite Hilbert space $\mathcal{H}_c \otimes \mathcal{H}_p$.

\subsection{Time-evolution}
A single step of quantum walks starts with applying the coin operator $C$ on the coin space. We first choose Grover's operator $G$ as the coin operator. The three-dimensional form of $G$ reads~\cite{Grover1996}
$$G=\frac{1}{3}\left[\begin{array}{ccc}-1 & 2 & 2 \\2 & -1 & 2 \\2 & 2 & -1\end{array}\right].$$\\
Then, the particle is moved by the conditional shift operator $S$, expressed as~\cite{Abal2010}
\begin{equation}
S=\sum_{j,s,\vec{r}}{\ket{j}\bra{j}\otimes\ketp{\vec{r}-(-1)^s\vec{v}_j,s\oplus1}\bra{\vec{r},s}},
\end{equation}
where
$$\vec{v}_1=(0,0),\, \vec{v}_2=(1,0),\, \vec{v}_3=(0,1).$$
Afterward the next step starts. During a single step, the position type $s$ is changed by the operator $S$, so the particle will move in opposite directions in the following step. Therefore this is a kind of ``flip-flop'' shift~\cite{Ambainis2005}.\par
The time-evolution operator $U$ of the total system is $U=S\cdot(C\otimes \mathbb{I}_p)$ where $\mathbb{I}_p$ is the identity operator acting on the position space. Accordingly the system's state $\ket{\Psi_t}$ after $t$ steps is derived by applying $U$ on the initial state for $t$ times:
\begin{equation}
\ket{\Psi_t}=(U)^t\ket{\Psi_0}.
\end{equation}\par
The state vector $\ket{\Psi_t}$ can be expanded in terms of all the coin and position eigenkets as
\begin{equation*}
\ket{\Psi_t}=\sum_{j, \vec{r},s}\psi_{t}^{j}(\vec{r}, s)\ket{j}\otimes\ketp{\vec{r},s}.
\end{equation*}
We introduce the vector
\begin{equation}
\ket{\psi_t(\vec{r},s)}=[\psi_{t}^{1}(\vec{r}, s),\psi_{t}^{2}(\vec{r}, s),\psi_{t}^{3}(\vec{r}, s)]^{\mathrm T},
\end{equation}
where the superscript $\mathrm T$ means transpose. The probability of finding the particle at the position $(\vec{r}, s)$ equals the squared norm of $\ket{\psi_t(\vec{r},s)}$ if its position is measured. \par
Taking into consideration the topology of the honeycomb network, we can obtain the recursive relations
\iffalse
\begin{widetext}
\begin{subequations}
\label{recursiveRelationsInRealSpace}
\begin{eqnarray}
\ket{\psi_{t+1}(x,y,0)}=C_1\ket{\psi_t(x,y,1)}+C_2\ket{\psi_t(x-1,y,1)}+C_3\ket{\psi_t(x,y-1,1)},\\
\ket{\psi_{t+1}(x,y,1)}=C_1\ket{\psi_t(x,y,0)}+C_2\ket{\psi_t(x+1,y,0)}+C_3\ket{\psi_t(x,y+1,0)},
\end{eqnarray}
\end{subequations}
\end{widetext}\fi
\begin{equation}\label{recursiveRelationsInRealSpace}
\ket{\psi_{t+1}(\vec{r},s)}=\sum_{j=1}^3C_j\ket{\psi_{t}(\vec{r}-(-1)^s\vec{v}_{j},s\oplus1)}
\end{equation}
where $C_j$ is a matrix whose $j$th row is the same as the $j$th row of the coin operator $C$ and other elements are all zeros, e.g.
\begin{equation*}
C_2=\frac{1}{3}\left[\begin{array}{ccc}0 & 0 & 0 \\2 & -1 & 2 \\0 & 0 & 0\end{array}\right].
\end{equation*}
Applying discrete Fourier transformation (DFT, see Appendix A) to two sides of Eq.~\eqref{recursiveRelationsInRealSpace}, we get the recursive relations in Fourier space:
\begin{subequations}
\label{recursiveRelationsInKLSpace}
\begin{eqnarray}
\ket{\tilde\psi_{t+1}(k,l,0)}&=&MC\ket{\tilde\psi_t(k,l,1)},\\
\ket{\tilde\psi_{t+1}(k,l,1)}&=&M^\dagger C\ket{\tilde\psi_t(k,l,0)},
\end{eqnarray}
\end{subequations}
where $M=\text{Diag}\{1,e^{-ik},e^{-il}\}$ and $k$ and $l$ are wave numbers corresponding to coordinates $x$ and $y$. \par
Assume the particle starts from position $(0,0,0)$ and denote the initial state of the coin $\alpha\ket{1}+\beta\ket{2}+\gamma\ket{3}$ as $[\alpha,\beta,\gamma]^{\mathrm T}$. The initial state of the system is therefore
\begin{equation}
\ket{\Psi_0}=[\alpha,\beta,\gamma]^{\mathrm T}\otimes\ketp{0,0,0},
\end{equation}
i.e., $$\ket{\psi_0(x,y,s)}=\delta_{s,0}\delta_{y,0}\delta_{x,0}[\alpha,\beta,\gamma]^{\mathrm T}, \forall x,y\in\mathbb{Z}.$$
$\delta_{a,b}$ is the Kronecker's delta. After DFT,
\begin{equation}\label{tildepsi_0(k,l,0)}
\ket{\tilde\psi_{0}(k,l,s)}=\delta_{s,0}[\alpha,\beta,\gamma]^{\mathrm T}, \forall k,l \in (-\pi,\pi].
\end{equation}
According to Eq.~\eqref{recursiveRelationsInKLSpace},
\begin{equation}\label{UEvolution}
\ket{\tilde\psi_{t}(k,l,0)}=(MCM^\dagger C)^{t/2}\ket{\tilde\psi_{0}(k,l,0)}
\end{equation}
when $t$ is even. Thus we need to explore the properties of the evolution operator $\tilde U=MCM^\dagger C$.

\section{\label{Limit of the state ket}Limit of the state vector}
Since $\tilde U$ is unitary, its eigenvalues can be denoted as $e^{i\xi_j}$. It is not straightforward to show that $\xi_1=0$ and $\xi_2=-\xi_3=\theta$
%\begin{equation}
%\xi_j=\begin{cases} 0&j=1 \\ \theta & j=2 \\ -\theta & j=3  \end{cases},
%\end{equation}
where $\theta$ satisfies
\begin{equation}\label{costheta}
\cos \theta = \frac{1}{9}[4\cos k+4\cos l+4\cos(k-l)-3].
%\sin \theta =  \sqrt{\frac{(-4 \cos (k-l)-4 \cos (k)-4 \cos (l)+3)^2}{81}-1}
\end{equation}
Eq.~\eqref{costheta} essentially agrees with the results in Refs.~\cite{machida2015limit,Higuchi20144197}. The corresponding eigenvectors are
\begin{widetext}
\begin{equation}\label{eigenvector}
\ket{\phi_j}=\frac{1}{N_j}
\left[
\begin{array}{c}
e^{ i(l-k+\xi_j ) }+e^{ i(k-l+\xi_j ) }+e^{ i(\xi_j -k) }+e^{ i(\xi_j -l) }-e^{ -ik }-e^{ -il }+\frac { 1   }{ 2 }e^{ i\xi_j} -\frac { 9 }{ 4 } e^{ 2i\xi_j  }-\frac { 1 }{ 4 }  \\
-e^{ i(l-k+\xi_j ) }+\frac { 1 }{ 2 }e^{ i(\xi_j -k) } +\frac { 1 }{ 2 }e^{ -ik } -e^{ -il }+\frac{1}{2}e^{ i\xi_j}+\frac { 1 }{ 2 }  \\
-e^{ i(k-l+\xi_j ) }+\frac { 1}{ 2 }e^{ i(\xi_j -l) }  +\frac {1}{ 2 } e^{ -il }  -e^{ -ik }+\frac { 1 }{ 2 }e^{ i\xi_j  } +\frac { 1 }{ 2 }
 \end{array}
 \right],
\end{equation}
\end{widetext}
where $N_j$ is a normalization factor.\par
To determine whether the localization will occur, one needs to calculate $P_{\infty}(0,0,0)$, i.e., the probability for the particle to be found at its starting point $(0,0,0)$ in the limit $t\rightarrow\infty$:
\begin{equation}\label{P(0,0,0)}
P_{\infty}(0,0,0)=\lim_{t\rightarrow\infty}\inp{\psi_{t}(0,0,0)}{\psi_{t}(0,0,0)}.
\end{equation}
Obviously, $P_{t}(0,0,0)=0, \forall \text{ odd } t$. As a result, we only need to pay attention to even-$t$ cases.\par
Substitute the results above into Eq.~\eqref{UEvolution},
\begin{equation}\label{psi_{t}(k,l,0)}
\ket{\tilde\psi_{t}(k,l,0)}=\sum_{j=1}^{3}e^{i\xi_j t/2}\ket{\phi_j}\inp{\phi_j}{\tilde\psi_{0}(k,l,0)}.
\end{equation}
Then the state ket in real space is derived by applying inverse DFT to Eq.~\eqref{psi_{t}(k,l,0)},
\begin{equation}\label{psy(x,y,0)}
\ket{\psi_t(x,y,0)}=\frac1{(2\pi)^2}\int_{V_2}\dif k\dif l e^{ikx+ily}\ket{\tilde\psi_{t}(k,l,0)},
\end{equation}
specially,
\begin{equation}\label{psy(0,0,0)}
\ket{\psi_t(0,0,0)}=\frac1{(2\pi)^2}\sum_{j=1}^{3}\int_{V_2}\dif k\dif l e^{i\xi_j t/2}\ket{\phi_j}\inp{\phi_j}{\tilde\psi_{0}(k,l,0)}.
\end{equation}\par
To further simplify the expression of $\ket{\psi_t(0,0,0)}$, we first introduce a mathematical lemma.
\begin{lemma}\label{th1}
Suppose $\varphi(x,y)\in L^1([a,b] \times [c,d])$ and it satisfyies that $a.e.\, x\in[a,b]$,
$\varphi_x(y) = \varphi(x,y)$ is a countably piecewise monotonic function of y, and in each monotonic interval $(c_n,d_n)$, $\varphi_x(y) \in C^2_B( (c_n,d_n))$ and the set $\{y\, |\, \partial_y\varphi(x,y) = 0\}$ is countable.
Let $f(x,y) = g(x,y)+i h(x,y)$ and $g,h \in L^1([a,b]\times [c,d])$. Then
 \begin{equation*}
   \lim\limits_{t \rightarrow \infty} \int_a ^b \int_c ^d f(x,y) e^{i  \varphi (x,y) t } \, \dif{y}\dif{x} = 0.
 \end{equation*}
\end{lemma}
The explanations of some notations as well as the detailed proof are given in Appendix B. According to Lemma~\ref{th1}, we deduce that the contributions to $\ket{\psi_t(0,0,0)}$ from items with $j = 2, 3$ in Eq.~\eqref{psy(0,0,0)} are negligible when $t$ is large enough. This is not hard to understand because when $t$ is very large, $e^{i\xi_j t/2}$ will oscillate much faster than $\ket{\phi_j}\inp{\phi_j}{\tilde\psi_{0}(k,l,0)}$ as $\xi_j$ changes. As a result, the integral will become zero when $t$ tends to infinity and we only need to consider the term with $j=1$. Since $\xi_1=0$,
\begin{equation}\label{psyinfty(0,0,0)}
\ket{\psi_\infty(0,0,0)}\sim\frac1{(2\pi)^2}\int_{V_2}\dif k\dif l \ket{\phi_1}\inp{\phi_1}{\tilde\psi_{0}(k,l,0)}.
\end{equation}\par
In Eq.~\eqref{psyinfty(0,0,0)}, only $\ket{\phi_1}$ is a function of $k,l$ while $\ket{\tilde\psi_{0}(k,l,0)}$ is constant according to Eq.~\eqref{tildepsi_0(k,l,0)}. Consequently Eq.~\eqref{psyinfty(0,0,0)} could be interpreted as a linear map from the initial state $\ket{\tilde\psi_{0}(k,l,0)}$ to $\ket{\psi_{\infty}(0,0,0)}$. The map is represented by a transformation matrix $F$
\begin{equation}\label{F}
F=\frac1{(2\pi)^2}\int_{V_2}\dif k\dif l \ket{\phi_1}\bra{\phi_1}.
\end{equation}
The matrix $F$ contains all the information about the walker's behavior at its starting point when $t\rightarrow\infty$. The existence of localization is directly related to the system's initial state via the matrix $F$. The final form of $F$ is obtained by exploiting the results in Eq.~\eqref{eigenvector}:
\begin{equation}\label{Fmatrix}
F=\frac{1}{6}
\begin{bmatrix}
2 & -1 & -1 \\
-1 & 2 & -1 \\
-1 & -1 & 2
\end{bmatrix}.
\end{equation}
Now we take a numerical example. Assuming $[\alpha,\beta,\gamma]^{\mathrm T}=[1,0,0]^{\mathrm T}$, which means the initial state of the coin is ``moving along $\vec{E_1}$", the probability of finding the particle at the position $(0,0,0)$ after a great number of steps is $|F[1,0,0]^{\mathrm T}|^2=1/6$ accordingly. Apparently, the probability is not zero and the localization is revealed. \par

We also conduct numerical simulations and the results are illustrated in FIG.~\ref{initialState100}. A noticeable peak arises in the left of FIG.~\ref{initialState100}. The right figure manifests the probability $P_t(0,0,0)$ oscillates around our theoretical value $P_{\infty}=1/6$ and clearly exhibits a tendency to converge, indicating the walker does have a nonzero probability to be localized at its starting point. \par

\begin{figure*}
\includegraphics[width=0.9\textwidth]{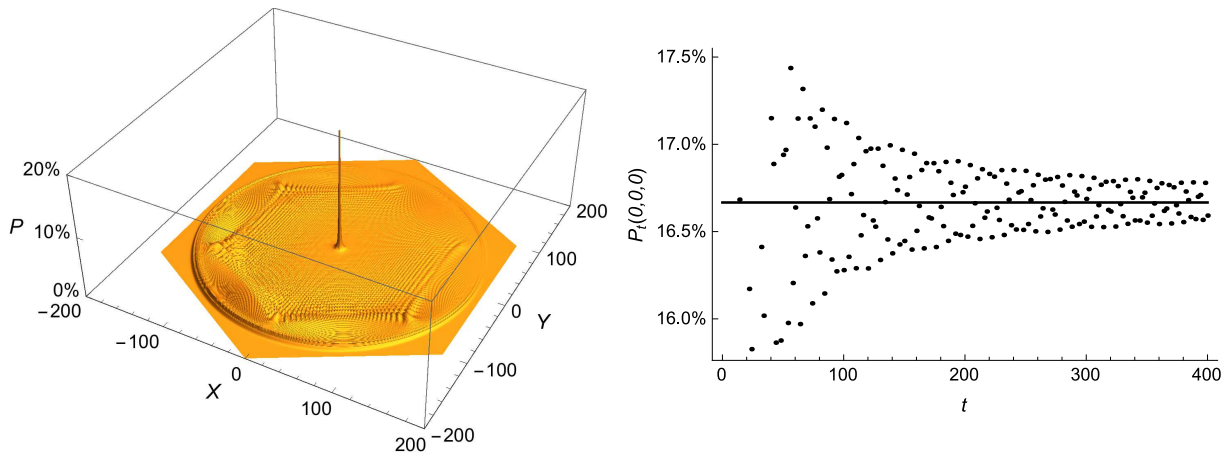}
\caption{\label{initialState100}(Color online) Numerical simulation results for the Grover walk on a honeycomb network with the initial state $\alpha=1, \gamma=\beta=0$. The left figure shows the particle's spatial probability distribution on the honeycomb network after 400 steps. Points in the right figure are the probability for finding the walker at its starting point as a function of time $t$ and the horizontal line is the theoretical result $P_{\infty}=1/6$.}
\end{figure*}\par

\section{Conditions for Localization}\label{Conditions and Dependence of Localization}

\subsection{Dependence on the initial state}\label{Dependence on the initial state}

We have seen that the limit behavior of the Grover walk is represented by Eq.~\eqref{psyinfty(0,0,0)}. The probability $P_{\infty}(0,0,0)$ is dependent on the coin's initial state $[\alpha,\beta,\gamma]^{\mathrm T}$ and it may decrease to zero under some circumstances. To verify this point, we work out the eigenvectors of $F$:
\begin{equation*}
\frac{1}{\sqrt{3}}\left[\begin{array}{c}1 \\1 \\1\end{array}\right],  \frac{1}{\sqrt 2}\left[\begin{array}{c}-1 \\0 \\1\end{array}\right] \text{ and } \frac{1}{\sqrt{6}}\left[\begin{array}{c}-1 \\2 \\-1\end{array}\right].
\end{equation*}
The corresponding eigenvalues are $0$, $1/2$ and $1/2$ respectively. So, if the initial coin state is $\frac1{\sqrt{3}}[1,1,1]^{\mathrm T}$, the uniform superposition state of the three eigenstates without phase differences, named the Grover state ~\cite{PhysRevA.78.032306}, localization will disappear. Numerical results are illustrated in the left of FIG.~\ref{minAndMax} and $P_t$ clearly converges to zero. It also manifests the Grover state is the sole initial state causing the walker to be delocalized. This result is consistent with the conclusion in Ref.~\cite{Higuchi20144197}. \par

On the other hand, the maximal value of $P_{\infty}(0,0,0)$ is $|1/2|^2=25\%$. Any initial coin state which is a linear combination of $\frac1{\sqrt{2}}[-1,0,1]^{\mathrm T}$ and $\frac1{\sqrt{6}}[-1,2,-1]^{\mathrm T}$ will lead to the maximum. For example, when the initial state $[\alpha,\beta,\gamma]^{\mathrm T}=\frac1{\sqrt{3}}[1,e^{i2\pi/3},e^{i4\pi/3}]^{\mathrm T}$, $P_{\infty}(0,0,0)$ will reach the maximal probability $25\%$ of localization (FIG.~\ref{minAndMax}). This is totally different from the Grover state above. \par
\begin{figure*}
\includegraphics[width=0.9\textwidth]{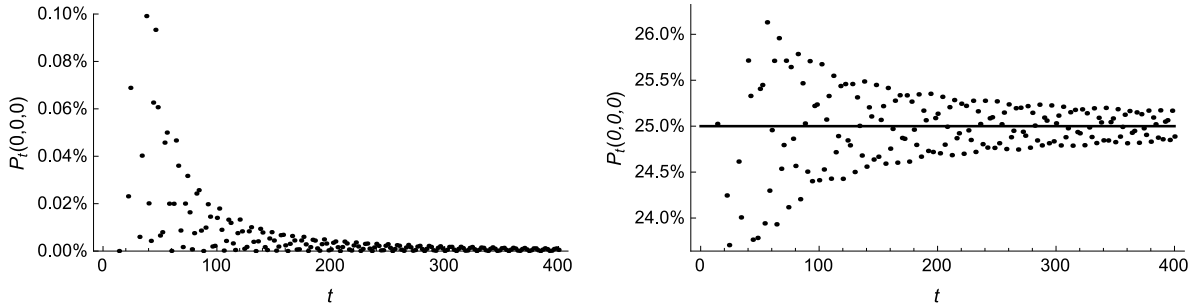}
\caption{\label{minAndMax}$P_{min}$ vs $P_{max}$. In the left figure, localization disappears when the initial state is the Grover state and in the right figure $P_{\infty}(0,0,0)$ reaches the maximal value with the initial coin state $\frac1{\sqrt{3}}[1,e^{i2\pi/3},e^{i4\pi/3}]^{\mathrm T}$.}
\end{figure*}

\subsection{Average localization probability}\label{average localization probability}
Since the probability $P_{\infty}(0,0,0)$ is dependent on the initial state of the coin, what is the average probability if the initial state is generated randomly? A random pure state of the coin can be obtained by applying a random unitary operator $U_3$ $\in$ SU(3) on an arbitrary coin state, such as $\ket{3}$ (see Sec.~\ref{twoSpaces}). The average probability is reasonably defined as
\begin{equation}\label{definitionOfPBar}
\overline{P}=\frac{1}{\Omega_3}\int{P_{\infty;\alpha,\beta,\gamma}\dif \Omega_3}=\frac{1}{\Omega_3}{\int{|FU_3\ket{3}|^2\dif \Omega_3}},
\end{equation}
where $\dif \Omega_3$ is the Haar measure of the group SU(3) and $\Omega_3=\int{\dif \Omega_3}$, standing for the volume of SU(3). A unitary operator $U_3$ is usually expressed with the generators $\lambda_j$ of SU(3) and the angular parameters $\zeta_j$ as~\cite{Tilma2002Generalized}
\begin{equation*}
U_3= e^{i\lambda_3\zeta_1}e^{i\lambda_2\zeta_2}e^{i\lambda_3\zeta_3}e^{i\lambda_5\zeta_4}e^{i\lambda_3\zeta_5}e^{i\lambda_2\zeta_6}e^{i\lambda_3\zeta_7}e^{i\lambda_8\zeta_8}.
\end{equation*}
where $0 \leqslant \zeta_1,\zeta_5 \leqslant \pi, 0\leqslant \zeta_3,\zeta_7\leqslant 2\pi, 0\leqslant \zeta_2,\zeta_4,\zeta_6\leqslant\frac{\pi}2$ and $0\leqslant\zeta_8\leqslant\sqrt 3 \pi$. Under this representation, the Haar measure $\dif \Omega_3$ of SU(3) and the random initial state $U_3\ket{3}$ read~\cite{Byrd98}
$$\dif \Omega_3=\sin 2\zeta_2\cos\zeta_4\sin^3\zeta_4\sin 2\zeta_6\prod_{j=1}^8 \dif\zeta_j.$$
\begin{equation}
U_3\ket{3}=e^{-2i\zeta_8/\sqrt{3}}
\begin{bmatrix}
e^{i(\zeta_1+\zeta_3)}\cos\zeta_2\sin\zeta_4\\
-e^{-i(\zeta_1-\zeta_3)}\sin\zeta_2\sin\zeta_4\\
\cos\zeta_4
\end{bmatrix}.
\end{equation}
Thus the average probability is obtained by finishing the integral in Eq.~\eqref{definitionOfPBar}. The result is
$$\overline{P}=\frac{1}{6}.$$\par
The result manifests that although a Grover state coin will delocalize the quantum walker, the expected probability $P_{\infty}$ of localization for a random initial state $\ket{\Psi_0}$ is still up to $16.7\%$ if we measure the walker’s position after even-t steps. Thus the localization phenomenon happens very frequently in our model and it could be useful in designing some quantum algorithms.
\par
\subsection{Dependence on coin operator}\label{Dependence on Coin Operator}
Apart from the initial coin state, it is obvious that the crucial point for the localization to occur is the existence of highly degenerate eigenvalue $1$ of the operator $\tilde U$, which is in agreement with literature~\cite{PhysRevE.72.056112,PhysRevA.69.052323}. This characteristic is undoubtedly determined by Grover's operator $G$. If there is another coin operator $C$ leading to the eigenvalue $1$ of $\tilde U$, the other two eigenvalues should also be like the form $e^{\pm i\theta}$ where $\theta$ is a function of $k$ and $l$. According to Lemma~\ref{th1}, these two terms in Eq.~\eqref{psy(0,0,0)} usually contribute less to the system's limit behavior, and the term with eigenvalue 1 may give rise to localization. Then what are the requirements for such a coin operator? We find a sufficient condition. \par

First, we extend the model of quantum walks on a honeycomb network in Sec.\ref{QWframework}. In the Grover walk, the operator $C$ is independent of position. We now release this restriction and let the coin operator be related to the position type. The time-evolution operator becomes
\begin{equation}
U=S\cdot(C\otimes \sum_{\vec{r}}\ket{\vec{r},0}\bra{\vec{r},0}+D\otimes \sum_{\vec{r}}\ket{\vec{r},1}\bra{\vec{r},1}).
\end{equation}
Consequently, the evolution operator $\tilde U$ in Fourier space is modified,
\begin{equation}
\tilde U=MDM^\dagger C.
\end{equation}\par
\paragraph*{\bf The sufficient condition:} If $D=C^\dagger$ and the eigenvectors of \/$C$ are all real vectors (overall phases are omitted), $1$ will be one of $\tilde U$'s eigenvalues and the localization is still present. \par
The proof is given in Appendix C. The Grover walk we have studied above is a special case since $C=G=G^\dagger=D$. Here we also give another example of the coin operator
\begin{equation}
H=\frac1{\sqrt{3}}
\left[
\begin{array}{ccc}
 1 & 1 & 1 \\
 1 & e^{i\frac{2 }{3}\pi} & e^{i\frac{ 4}{3}\pi} \\
 1 & e^{i\frac{4}{3}\pi } & e^{i\frac{2}{3}\pi} \\
\end{array}
\right].
\end{equation}
$H$ looks much more like the three-dimensional version of Hadamard matrix and $H$ satisfies all the conditions above. We take $C=H$ and $D=H^\dagger$ with an initial coin state $[1,0,0]^{\mathrm T}$ and perform numerical simulations. The result is given in FIG.~\ref{H(1,0,0)}, which obviously demonstrates that the walker has a nonzero probability to stay at its starting point when $t\rightarrow\infty$.\par
\begin{figure*}
\includegraphics[width=0.9\textwidth]{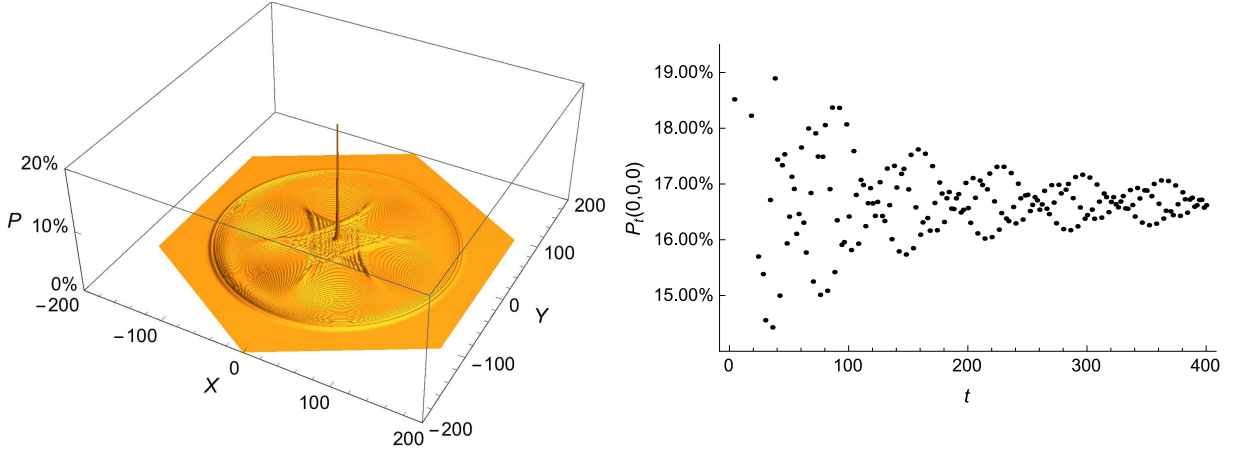}
\caption{\label{H(1,0,0)}(Color online) Localization exists when $C=H$ and $D=H^\dagger$ with the initial state $\alpha=1$ and $\beta=\gamma=0$.}
\end{figure*}
Specially if we additionally require all the eigenvalues of $C$ to be $\pm 1$, then we will get the sufficient condition for localization when the coin operator is independent of positions.
\section{\label{4DCondition}Four-state quantum walks on a honeycomb network}
\subsection{Definition and transformation matrices}
In this section, we consider another modified model in which the dimension of the coin space is enlarged to four. The additional fourth component represents not hopping, which means the walker is allowed to stay at its current position without moving. As a result, the shift operator now reads
\begin{equation*}
S=\ket{4}\bra{4}\otimes\mathbb{I}_p+\sum_{j=1}^3\sum_{s,\vec{r}}{\ket{j}\bra{j}\otimes\ketp{\vec{r}-(-1)^s\vec{v}_j,s\oplus1}\bra{\vec{r},s}},
\end{equation*}
and the revised  recursive relation is given as
$$\ket{\psi_{t+1}(\vec{r},s)}=C_4\ket{\psi_t(\vec{r},s)}+\sum_{j=1}^3C_j\ket{\psi_{t}(\vec{r}-(-1)^s\vec{v}_{j},s\oplus1)},$$
where the coin operator $C$ now should be the four-dimensional Grover's operator G~\cite{Grover1996}
\begin{equation}
G=\frac{1}{2}\begin{bmatrix}
-1 & 1 & 1 &1\\
1 & -1 & 1 &1\\
1 & 1 & -1 &1\\
1 & 1 & 1 &-1\\
\end{bmatrix}.
\end{equation}
The fourth coin basis state $\ket{4}$ stands for ``not moving''. Similarly, the recursive relations in Fourier space Eq.~\eqref{recursiveRelationsInKLSpace} also change to
\begin{eqnarray*}
\ket{\tilde\psi_{t+1}(k,l,0)}&=&C_4\ket{\psi_t(k,l,0)}+M'C\ket{\tilde\psi_t(k,l,1)},\\
\ket{\tilde\psi_{t+1}(k,l,1)}&=&C_4\ket{\psi_t(k,l,1)}+M'^\dagger C\ket{\tilde\psi_t(k,l,0)},
\end{eqnarray*}
with $M'=\text{Diag}\{1,e^{-ik},e^{-il},0\}$, which can be further expressed as
$$\begin{bmatrix}
\ket{\tilde\psi_{t+1}(k,l,0)}\\
\ket{\tilde\psi_{t+1}(k,l,1)}
\end{bmatrix}=
\begin{bmatrix}
C_4 & M'C\\
M'^\dagger C^\dagger & C_4
\end{bmatrix}
\begin{bmatrix}
\ket{\tilde\psi_{t}(k,l,0)}\\
\ket{\tilde\psi_{t}(k,l,1)}
\end{bmatrix}.$$
Therefore,
\begin{equation}
\begin{bmatrix}
\ket{\tilde\psi_{t}(k,l,0)}\\
\ket{\tilde\psi_{t}(k,l,1)}
\end{bmatrix}=(\tilde U')^t
\begin{bmatrix}
\ket{\tilde\psi_{0}(k,l,0)}\\
\ket{\tilde\psi_{0}(k,l,1)}
\end{bmatrix},
\end{equation}
where $\tilde U'=\begin{bmatrix}
C_4 & M'C\\
M'C^\dagger & C_4
\end{bmatrix}$ is an $8 \times 8$ matrix. So we start to examine the properties of the evolution operator $\tilde U'$.\par
The eigenvalues of $\tilde U'$ are $1$, $-1$, $e^{\pm i\theta_1}$ and $e^{\pm i\theta_2}$ where $\theta_1$ and $\theta_2$ are functions of $k$ and $l$. According to Lemma~\ref{th1}, when the time $t$ is large enough, we only need to pay attention to eigenvalues $\pm 1$ and their corresponding eigenvectors. The eigenvectors are
\begin{equation*}
\begin{split}
\ket{\phi'_1}=&[e^{i l}-e^{i k},1-e^{i l},-1+e^{i k},0,e^{i k}-e^{i l},\\
&e^{i k} (-1+e^{i l}),-e^{i l}(-1+e^{i k}),0]^{\mathrm T}
\end{split}
\end{equation*}
corresponding to $1$ and
\begin{eqnarray*}
\ket{\phi'_2}&=&[-1,0,0,1,-1,0,0,1]^{\mathrm T}\\
\ket{\phi'_3}&=&[0,e^{-i k},0,-e^{-i k},0,1,0,-1]^{\mathrm T}\\
\ket{\phi'_4}&=&[0,0,e^{-i l},-e^{-i l},0,0,1,-1]^{\mathrm T}
\end{eqnarray*}
corresponding to $-1$. \par
Now, if the walker still starts from the position $(0,0,0)$ with an initial coin state $[\alpha,\beta,\gamma,\mu]^{\mathrm T}$, the initial state of the overall system reads
$$\ket{\Psi_0}=[\alpha,\beta,\gamma,\mu]^{\mathrm T}\otimes\ketp{0,0,0}$$
and in Fourier space
$$\ket{\tilde\psi_{0}(k,l,s)}=\delta_{s,0}[\alpha,\beta,\gamma,\mu]^{\mathrm T}, \forall k,l \in (-\pi,\pi].$$
Since $\tilde U$ has two highly degenerate eigenvalues $\pm 1$, we denote
\begin{eqnarray}
\label{mathcalF+}\mathcal{F}_+&=&\frac1{(2\pi)^2}\int_{V_2}\dif k\dif l \ket{\phi_1}\bra{\phi_1},\\
\label{mathcalF-}\mathcal{F}_-&=&\frac1{(2\pi)^2}\sum_{j=2}^{4}\int_{V_2}\dif k\dif l \ket{\phi_j}\bra{\phi_j},
\end{eqnarray}
where vectors $\ket{\phi_j}$ are orthogonalized and normalized forms of $\ket{\phi'_j}$.
Hence, analogous to Eq.~\eqref{psyinfty(0,0,0)}, for a large $t$ we have
\begin{equation}
\begin{bmatrix}
\ket{\psi_{t}(0,0,0)}\\
\ket{\psi_{t}(0,0,1)}
\end{bmatrix}\sim (\mathcal{F}_+ +(-1)^t\mathcal{F}_-)
\begin{bmatrix}
\ket{\tilde\psi_{0}(k,l,0)}\\
0
\end{bmatrix},
\end{equation}
or
\begin{equation}\label{psyt(0,0,0)4}
\ket{\psi_{t}(0,0,0)}\sim (F_+ +(-1)^tF_-)\ket{\tilde\psi_{0}(k,l,0)},
\end{equation}
where the elements of $4 \times 4$ matrices $F_\pm$ are
$$(F_\pm)_{ij}=(\mathcal{F}_\pm)_{ij}, 1\leqslant i,j\leqslant4.$$\par
Different from the three-state walk, the extra freedom causes the probability $P_t(0,0,0)$ nonzero when $t$ is odd. However, Eq.~\eqref{psyt(0,0,0)4} implies the transformation matrices for odd and even $t$ are different, due to the negative eigenvalue $-1$, i.e.
\begin{eqnarray}
P_{\infty,e}(0,0,0)&=& |F_{e}\ket{\tilde\psi_{0}(k,l,0)}|^2\\
P_{\infty,o}(0,0,0)&=& |F_{o}\ket{\tilde\psi_{0}(k,l,0)}|^2
\end{eqnarray}
where $F_{e}=F_++F_-$ and $F_{o}=F_+-F_-$. Subscripts `e' and `o' stand for `even' and `odd' respectively.\par
On account of the complexity of Eq.~\eqref{mathcalF-}, we only get a numerical result of $F_-$:
\begin{equation*}
F_-\approx\begin{bmatrix}
 0.320303 & -0.070303 & -0.070303 & -0.179697 \\
 -0.070303 & 0.320303 & -0.070303 & -0.179697 \\
 -0.070303 & -0.070303 & 0.320303 & -0.179697 \\
 -0.179697 & -0.179697 & -0.179697 & 0.539090 \\
\end{bmatrix},
\end{equation*}
while the accurate form of $F_+$ is obtained
\begin{equation}
F_+=\dfrac{1}{12}
\begin{bmatrix}
2 & -1 & -1 & 0\\
-1 & 2 & -1 & 0\\
-1 & -1 & 2 & 0\\
0 & 0 & 0 & 0\\
\end{bmatrix}.
\end{equation}\par
As an example, with the initial coin state set to $\frac{1}{\sqrt{2}}[1,1,0,0]^{\mathrm{T}}$, the value of $P_t(0,0,0)$ would converge to two values
\begin{eqnarray}
P_{\infty,e}&=&|F_{e}\frac{1}{\sqrt{2}}[1,1,0,0]^{\mathrm{T}}|^2\approx 0.222901\\
P_{\infty,o}&=&|F_{o}\frac{1}{\sqrt{2}}[1,1,0,0]^{\mathrm{T}}|^2\approx 0.092699
\end{eqnarray}
according to the derivations above. To confirm the results, numerical simulations are performed, illustrated in Fig.~\eqref{G41000}.
\begin{figure}
\includegraphics[width=0.45\textwidth]{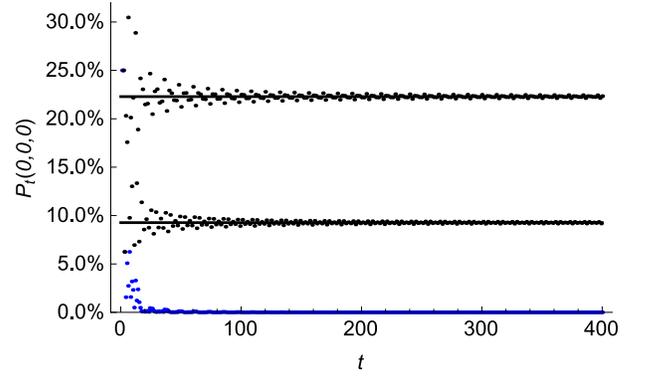}
\caption{\label{G41000}(Color online) Simulation results for four-state quantum walks with initial state $\frac{1}{\sqrt{2}}[1,1,0,0]^{\mathrm{T}}$ (black) and $\frac{1}{2}[1,1,1,1]^{\mathrm{T}}$ (blue).}
\end{figure}
The data points do distribute in two branches and get close to our analytical predictions (the two horizontal lines) as the step number $t$ gets large.
\subsection{Maximal, minimal and average probability of localization}
Once we have obtained transformation matrices, the properties of this quantum walk are easily analyzed. The eigenvalues of $F_{e}$ are $0.718787$, $0.640606$, $0.640606$, and zero. Those of $F_{o}$ are $-0.718787$, $-0.140606$, $-0.140606$ and zero, intimating that $P_o$ and $P_e$ have equal maximal values, about $|\pm 0.718787|^2\approx 0.516655$, as well as the same minimum zero. The eigenvectors of $F_{o}$ and $F_{e}$ corresponding to zero are both $\frac{1}{2}[1,1,1,1]^{\mathrm T}$. So the Grover state is also the sole initial coin state that will delocalize the particle in four-state walks, which is shown in Fig.~\ref{G41000}.\par
What is more interesting is that the eigenvector of $F_{e}$ corresponding to 0.718787 is also the same as that of $F_{o}$ corresponding to $-0.718787$, whose numerical value reads
\begin{equation*}
\begin{split}
\ket{\varphi_{max}}\approx &[-0.288675, -0.288675, -0.288675, 0.866025]^{\mathrm{T}}\\
\approx &\frac{1}{2\sqrt{3}}[-1,-1,-1,3]^{\mathrm{T}}.
\end{split}
\end{equation*}
This suggests if the coin state is initialized with $\ket{\varphi_{max}}$, $P_{e}$ and $P_{o}$ will reach the maximal value simultaneously. The numerical results in Fig.~\ref{max4} do show that the probabilities of localization for even and odd times converge to a single limit, just as we have calculated. However, the ket $\ket{\Psi_t}$ still oscillates, for $F_{o}\ket{\varphi_{max}}$ is different from $F_{e}\ket{\varphi_{max}}$.\par
\begin{figure}
\includegraphics[width=0.45\textwidth]{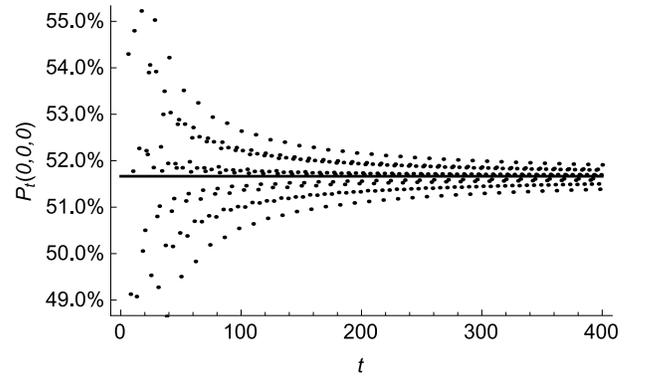}
\caption{\label{max4} When the initial coin state is $\alpha=\beta=\gamma=-\frac{1}{2\sqrt{3}}$ and $\mu=\frac{\sqrt{3}}{2}$, the values of $P_{e}$ and $P_{o}$ converge to the same limit, about 0.516655.}
\end{figure}
The maximum of $P_{\infty}$ here is around twice as much as that for the previous three-state coin model, indicating the localization is much stronger. We consider the average probability again to make a more appropriate comparison. In this case, arbitrary unitary operators $U_4\in$ SU(4) and the Haar measure $\dif \Omega_4$~\cite{Tilma2002Generalized} of SU(4) are involved:
\begin{eqnarray}
\overline{P_{e}}&=&\frac{1}{\Omega_4}{\int{|F_{e}U_4\ket{4}|^2\dif \Omega_4}}\approx 0.334352,\\
\overline{P_{o}}&=&\frac{1}{\Omega_4}{\int{|F_{o}U_4\ket{4}|^2\dif \Omega_4}}\approx 0.139049.
\end{eqnarray}
$\overline{P_{e}}$ is much larger than $\overline{P}=1/6$ in Sec.~\ref{average localization probability} and $\overline{P_{o}}$ is only slightly smaller than $\overline{P}$. In fact the probability $P_t(0, 0, 0)$ is always zero in the previous model when t is odd. Thus the localization in four-state quantum walks is much stronger than that in the previous model due to the additional fourth coin state of not hopping.
\section{\label{Comparison}Comparisons with previous works}
Now we compare our results of quantum walks on a honeycomb network to other kinds of quantum walks studied in the literature.\par
The first is work~\cite{PhysRevE.72.056112} by Inui \textsl{et al} about a particle walking on a line under the control of Grover's operator. Besides walking left and right, the particle has an additional freedom to stay at its current position. So the coin space is also three-dimensional. For the situation considered by Inui \textsl{et al}, we derive the corresponding transformation matrix $F'$
\begin{equation}
F'=\frac{1}{\sqrt{6}}\left[
\begin{array}{ccc}
 1 & \sqrt{6}-2 & 2\sqrt{6}-5 \\
 \sqrt{6}-2 & \sqrt{6}-2 & \sqrt{6}-2 \\
 2\sqrt{6}-5 & \sqrt{6}-2 & 1 \\
\end{array}
\right].
\end{equation}
Likewise, we calculate the average probability of this walk and get the result $\overline{P'}\approx 0.1684$, which is slightly larger than $1/6$ for the three-state Grover walk on a honeycomb network. However, we notice that the position space here is one dimensional but the honeycomb network is a 2D lattice. Meanwhile, in our case the walker does not have the extra freedom of ``stay". If the walker on a honeycomb network is also allowed not to hop, just as we calculated in Sec.~\ref{4DCondition}, the average probability will get much larger. Thus we can say the degree of localization is relatively high in the Grover walk on a honeycomb network. The particle can be found at its starting position with a good chance. \par
On the other hand, we can find out the maximal value of $P'_{\infty}$ in this 1D three-state walk by eigendecomposition of $F'$. The eigenvalues of $F'$ are zero, $\sqrt{6}-2$ and $3-\sqrt{6}$ and the corresponding eigenvectors are
\begin{equation*}
\frac{1}{\sqrt{6}}\begin{bmatrix}1\\-2\\1\end{bmatrix}, \frac{1}{\sqrt 2}\begin{bmatrix}-1\\0\\1\end{bmatrix} \text{ and } \frac{1}{\sqrt{3}}\begin{bmatrix}1\\1\\1\end{bmatrix}.
\end{equation*}
Hence the initial coin state $\frac{1}{\sqrt{6}}[1,-2,1]^{\mathrm T}$ will delocalize the walker, which agrees with the result in Ref.~\cite{PhysRevE.72.056112}. However, the maximum of $P'_{\infty}$ should be $|3-\sqrt{6}|^2\approx 0.303$ corresponding to the initial coin state $\frac{1}{\sqrt{3}}[1,1,1]^{\mathrm T}$, instead of the value $0.202$ claimed in Ref.~\cite{PhysRevE.72.056112}.  Actually, one can also get the maximal value 0.303 by submitting $\alpha=\beta=\gamma=1/\sqrt{3}$ into Eq. (8) of Ref.~\cite{PhysRevE.72.056112}. It is obvious that our mathematical method with the transformation matrix $F$ is very convenient in studying localization phenomena in quantum walks.\par
Apart from the localization probability, another remarkable difference is that in 1D three-state walks, the limit state converges to a static wave function, while on a honeycomb network, the limit ket oscillates between two values. Furthermore, the oscillation in four-state walks is not so palpable as that of three-state walks on the honeycomb network, since it results from the negative eigenvalue $-1$ of $\tilde U'$, the evolution operator in Fourier space.\par
Koll\'ar \textsl{et al}. in their previous work~\cite{PhysRevA.82.012303} investigated another model where the particle walks on a triangular lattice with a three-dimensional coin operator, similar to three-state quantum walks in this paper. The authors found when the coin operator is set to Grover's operator, the walker still has a rapidly decaying probability to be localized at the origin. Analyses suggest that coin operators leading to localization will transform quantum walks on a 2D plane to quantum walks on a quasi-one-dimensional line. In contrast, Grover walks on a honeycomb network undoubtedly allow the appearance of localization, and we have also found many more coin operators leading to walker's localization in the previous section. Besides, by comparing FIG.~\ref{H(1,0,0)} with FIG.~\ref{initialState100}, we can see that walker under the control of $H$ displays different probability distribution from $G$. There are nontrivial phenomena present in these models. So we can see quantum walks on a honeycomb network do have more different features.\par

The works~\cite{Higuchi20144197} by Y. Higuchi \textsl{et al}. and~\cite{machida2015limit} by Machida also analyzed the localization in quantum walks on a hexagonal lattice. The coin operators are Grover's operator $G$ and $C'$, respectively, in their articles, where
\begin{equation}
C^\prime=\frac{1}{2}
\begin{bmatrix}
-1-\cos\epsilon & \sqrt{2}\sin\epsilon & 1-\cos\epsilon \\
\sqrt{2}\sin\epsilon & 2\cos\epsilon & \sqrt{2}\sin\epsilon \\
1-\cos\epsilon & \sqrt{2}\sin\epsilon & -1-\cos\epsilon
\end{bmatrix}
\end{equation}
and the parameter $\epsilon\in [0,2\pi)$. In fact when $\epsilon$ is set to $\arcsin(-1/3)$, $C^\prime$ equals Grover's operator. Their results of long-time limit of the localization probability are congruent with ours. However, the method of oblique coordinate system used to label the vertices in our article is more intuitive and concise. We have conducted more detailed analysis such as the average probability and the dependence on the coin operator. Moreover, we can see that the eigenvectors of $C'$ are
\begin{equation}
\begin{bmatrix}
-1\\
0\\
1
\end{bmatrix} \text{, }
\begin{bmatrix}
\cot\epsilon+\csc\epsilon\\
-\sqrt{2}\\
\cot\epsilon+\csc\epsilon
\end{bmatrix} \text{ and }
\begin{bmatrix}
1-\cos\epsilon\\
\sqrt{2}\sin\epsilon\\
1-\cos\epsilon
\end{bmatrix},
\end{equation}
and $C^\prime={C^\prime}^\dagger$. So $C'$ satisfies the sufficient condition in Sec.~\ref{Dependence on Coin Operator}. Thus 1 is one of the the time-evolution operator's eigenvalue and localization is probable to occur in Machida's model according to the more general result in Sec.~\ref{Dependence on Coin Operator}. \par

As a by-product, we numerically estimate the spreading speed of the walker mentioned in their papers. Reasonably, we use the average radius $\bar r$ defined as
\begin{equation}
\bar r=\sum_{x,y,s}{r_{x,y,s}P_t(x,y,s)}
\end{equation}
from the starting point to describe how far the walker has reached, where
$$r_{x,y,s}=\sqrt{\frac{3}{4}(x+\frac{s}{3})^2+(y+\frac{x+s}{2})^2}$$ is the distance between position $(x,y,s)$ and $(0,0,0)$. The results for the circumstances of Fig.~\ref{initialState100} and Fig.~\ref{H(1,0,0)} are plotted in Fig.~\ref{rbar}, where the data points are distributed in two lines very well, i.e., the walker is linearly spreading on the network and the radius of the ring in Fig.~\ref{initialState100} and in Fig.~\ref{H(1,0,0)} are proportional to time, which agrees with Refs.~\cite{machida2015limit} and~\cite{Higuchi20144197}. Additionally, the difference between the slopes of the two lines indicates the spreading speed is related to the coin operator.
\begin{figure}
\includegraphics[width=0.4\textwidth]{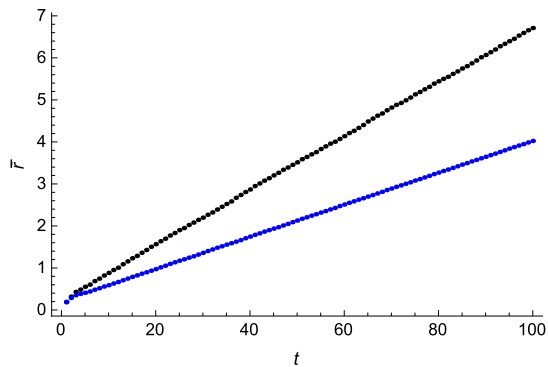}
\caption{\label{rbar}(Color online) Relationship between average position $\bar{r}$ and time $t$. Black and blue dots stand for Fig.~\ref{initialState100} and Fig.~\ref{H(1,0,0)}, respectively.}
\end{figure}

\section{\label{summary}Summary}
In this paper, we have established a precise mathematical description of quantum walks on a honeycomb network and conducted a comprehensive and rigorous study on the localization effect. We first set the coin operator to Grover's operator. With discrete Fourier transformation, we analytically obtained the system's state vector. The localization probability in the limit of large number of steps is connected to the initial coin state through a concise transformation matrix $F$, which completely contains all the information about the walker's limit behavior. We then analytically showed localization and delocalization. The results were all supported by numerical simulations. The average probability of localization's occurrence is $1/6$ if the system's initial state is randomly generated. Other coin operators were also discussed and we derived a sufficient condition for localization. Further, we extended the coin space to four and conducted parallel calculations, finding an oscillating limit state and deriving a larger probability. Some other models in literature were reviewed. Basing on the average probabilities, we concluded that the trapping effect of a honeycomb network is relatively strong. Our mathematical method with the transformation matrix is very convenient to study localization in quantum walks. We numerically showed the walker's linear spreading property as well.
\par

\iffalse
From discussions in the paper, we also see that the phase plays an important role in quantum walks and that is the most salient difference between classical and quantum dynamics. Phase factors affect the degree of localization through the eigenvalues of operator $\tilde U$, the initial state and eigenvectors of coin operator. In some previous experiments~\cite{PhysRevA.61.013410}, quantum walks were implemented with only classical optics devices utilized, since classical electromagnetic wave can well imitate walker's de Broglie wave. We therefore conclude that adapting phase differences and controlling interference relationship play essential roles in the appearance and disappearance of localization.\par
Quantum walks are often used to develop quantum algorithms. Taking account of the broad choice of coin operators and vast studies in existing hexagonal 2D materials like graphene and h-BN, the localization property of quantum walks on a honeycomb network suggests its potential application in designing and implementing new algorithms on these materials. We also notice that the simple Grover walk model can be extended by much more approaches, such as additional gates or vertex-dependent coin operators. Thus we may find higher degree of localization or other novel phenomena. There are plenty of open and fascinating questions.  \par
\fi

\begin{acknowledgments}
We wish to thank Sisi Zhou and Ru Cao for their assistance and Yafang Xu, Xingfei Zhou, Hongyi Zhang and Professor Guojun Jin for their advice. This work is supported by the National Natural Science Foundation of China under Grants No. 11275181 and No. 11475084,
and the Fundamental Research Funds for the Central Universities under Grant No. 20620140531.
\end{acknowledgments}

\appendix
\section{\label{DTFT}discrete Fourier transformation}
Suppose $X(\vec m)$ is a function of $\vec m$, an n-dimensional vector whose components $m_j\in \mathbb{Z}$. Its discrete Fourier transformation is defined as
\begin{equation}
 {\tilde X}(\vec k)=\sum_{m_j\in \mathbb{Z}} X(\vec m)e^{-i\vec k\cdot \vec m},
 \end{equation}
 where $\vec k$ is an n-dimensional wave vector with components $k_j\in (-\pi,\pi]$. The inverse transformation is:
 \begin{equation}
 X(\vec m)=\dfrac{1}{(2\pi)^n}\int_{V_n}{{\tilde X}(\vec k )e^{i\vec k\cdot \vec m}\dif \vec k}\label{DTFT-R}.
 \end{equation}
The integration region $V_n$ in Eq.~\eqref{DTFT-R} is $k_j\in(-\pi,\pi]$, representing an n-dimensional hypercube whose sides equal $2\pi$ in the Fourier space.

\section{Proof for Lemma~\ref{th1}}
Mathematical notations are as follows:
\begin{itemize}
\item $L^1([a,b])$: the set of Lebesgue integrable real functions defined on the interval $[a,b]$.
\item $C^2_B(A)$: the set of bounded real functions defined on set A with continuous first and second derivatives.
\item $\mathfrak{M} ([a,b])$: the set of Lebesgue measurable real functions defined on the interval $[a,b]$.
\item $P(x)$, a.e. $x\in A$: the proposition $P(x)$ does not hold only for $x\in A_1\subset A$ and the Lebesgue measure of $A_1$ is zero.
\end{itemize}

\begin{lemma}[Riemann-Lebesgue lemma]\label{RL lemma}
If $f(x) \in L^1([a,b])$, then
 \begin{equation*}
   \lim\limits_{t \rightarrow \infty} \int_a ^b f(x) \cos{\left(tx \right)} \, \dif{x} = \lim\limits_{t \rightarrow \infty} \int_a ^b f(x)\sin{\left(tx \right)} \, \dif{x} = 0.
 \end{equation*}
\end{lemma}

\begin{lemma}\label{lemma3}
If a countably piecewise monotonic function $\varphi(x)$ definied on $[a,b]$ satisfies that in each monotonic interval $(a_n,b_n)$, $\varphi(x) \in C^2_B((a_n,b_n))$ and the set $\{x\, |\, \varphi'(x) = 0\}$ is countable, then
 \begin{equation}
   \lim\limits_{t \rightarrow \infty} \int_a ^b \cos{\left(  \varphi (x) t \right)} \, \dif{x} = \lim\limits_{t \rightarrow \infty} \int_a ^b \sin{\left(  \varphi (x) t \right)} \, \dif{x} = 0.
 \end{equation}
\end{lemma}

\begin{proof}
We prove it in the case of cosine. The domain of $\varphi(x)$ can be divided into countable intervals, i.e.,
\begin{equation*}
   [a,b)=\bigcup_n [a_n,b_b)\text{ and }[a_n,b_n)\cap[a_m,b_m) = \varnothing, \text{ for } n \ne m,
\end{equation*}
and in each interval the function is monotonic. In this case, each piece is bounded and differentiable to the second order.

Consider one of those intervals $[a_n,b_n)$. First, based on the theorem of existence of inverse function, there is an inverse function of $\varphi$, called $\varphi^{-1}$. Besides, near the end points $a_n$ and $b_n$, we have
\begin{equation*}
   \lim\limits_{\epsilon \rightarrow 0} \int_{a_n} ^{a'_n} \left| \cos{\left(  \varphi (x) t \right)}\right| \, \dif{x} =\lim\limits_{\epsilon \rightarrow 0} \int_{b'_n} ^{b_n} \left| \cos{\left(  \varphi (x) t \right)}\right| \, \dif{x}  =0,
\end{equation*}
where $a'_n=a_n+\epsilon$ and $b'_n=b_n-\epsilon$. Because the set $\{x\, |\, \varphi'(x) = 0\}$ is countable,
$\varphi'(a'_n)$ or $\varphi'(b'_n)$ is not zero for almost all $\epsilon$. Now in the interval $[a'_n,b'_n)$, by making a substitution $u=\varphi(x)$ and integrating by parts, we get
\begin{widetext}
\begin{eqnarray}
   \int_{a'_n} ^{b'_n} \cos{(  \varphi (x) t )} \, \dif{x}&=&  \int_{\varphi(a'_n)} ^{\varphi(b'_n)} \frac{\cos{(ut)}}{\varphi'[\varphi^{-1}(u)]} \, \dif{u} = \frac{1}{t} \frac{\sin{(ut)}}{\varphi'[\varphi^{-1}(u)]} \bigg|_{\varphi(a'_n)}^{\varphi(b'_n)} + \frac{1}{t} \int_{\varphi(a'_n)} ^{\varphi(b'_n)}  \sin{(ut)} \frac{\varphi''[\varphi^{-1}(u)]}{\{\varphi'[\varphi^{-1}(u)]\}^3} \, \dif{u}\nonumber \\
   \label{IBP}&=& \frac{1}{t} \left(\frac{\sin{\left(\varphi(b'_n)t\right)}}{\varphi'(b'_n)}-\frac{\sin{\left(\varphi(a'_n)t\right)}}{\varphi'(a'_n)}\right)
   + \frac{1}{t} \int_{\varphi(a'_n)} ^{\varphi(b'_n)}  \sin{\left(  u t \right)} \frac{\varphi''\left[\varphi^{-1}(u)\right]}{\left\{\varphi'\left[\varphi^{-1}(u)\right]\right\}^3} \, \dif{u}.
\end{eqnarray}
\end{widetext}
Since $|\sin{\left(  u t \right)}|\leq 1$ and $\varphi'(a'_n) \ne 0,\varphi'(b'_n) \ne 0$, the first term tends zero when $t \rightarrow \infty$. For the second term, making inverse substitution again,
\begin{align}
\begin{split}
\int_{\varphi(a'_n)} ^{\varphi(b'_n)}  \frac{\varphi''\left[\varphi^{-1}(u)\right]}{\left\{\varphi'\left[\varphi^{-1}(u)\right]\right\}^3} \, \dif{u} = \int_{a'_n} ^{b'_n} \frac{\varphi''(x)}{\left[\varphi'(x)\right]^2} \, \dif{x}
\\ =\frac{1}{\varphi'(a'_n)}-\frac{1}{\varphi'(b'_n)}.
\end{split}
\end{align}
According to Lemma \ref{RL lemma}, the second term tends to zero when $t \rightarrow \infty$.\par
Now Let us see the integral in interval $[a_n,b_n)$
\begin{equation*}
    \int_{a_n} ^{b_n} \cos{\left(  \varphi (x) t \right)} \, \dif{x}
    =(\int_{a_n} ^{a'_n} +\int_{a_n'} ^{b_n'} +\int_{b'_n} ^{b_n}) \cos{\left(  \varphi (x) t \right)} \, \dif{x}.
\end{equation*}
After letting $\epsilon \rightarrow 0$ and $t \rightarrow \infty$,
\begin{equation*}
    \lim\limits_{t \rightarrow \infty} \int_{a_n} ^{b_n} \cos{\left(  \varphi (x) t \right)} \, \dif{x}=0.
\end{equation*}
Finally, since $\lim\limits_{t \rightarrow \infty} \int_{a} ^{b} \left| \cos{\left(  \varphi (x) t \right)} \right| \, \dif{x} < b-a$ is uniformly convergent, the limit and the sum is commutable. So,
\begin{equation*}
\begin{split}
    \lim\limits_{t \rightarrow \infty} \int_{a} ^{b} \cos{\left(  \varphi (x) t \right)} \, \dif{x}
    = \lim\limits_{t \rightarrow \infty} \sum_n \int_{a_n} ^{b_n} \cos{\left(  \varphi (x) t \right)} \, \dif{x}\\
    = \sum_n \lim\limits_{t \rightarrow \infty} \int_{a_n} ^{b_n} \cos{\left(  \varphi (x) t \right)} \, \dif{x}
    =0.
\end{split}
\end{equation*}
\end{proof}

\begin{lemma}\label{lemma4}
For a countably piecewise monotonic function $\varphi(x)$ satisfying that in each monotonic interval $(a_n,b_n)$, $\varphi(x) \in C^2_B((a_n,b_n))$ and the set $\{x\, |\, \varphi'(x) = 0\}$ is countable, if $f(x) = g(x)+i h(x)$ and $g,h \in L^1([a,b])$, then
 \begin{equation}
   \lim\limits_{t \rightarrow \infty} \int_a ^b f(x) e^{i  \varphi (x) t } \, \dif{x} = 0.
 \end{equation}
\end{lemma}

\begin{proof}
\begin{equation}
\begin{split}\label{f(x)exp(iphi(x))}
  f(x) e^{i  \varphi (x) t }=& g(x) \cos{\left(  \varphi (x) t \right)} -h(x) \sin{\left(  \varphi (x) t \right)}\\
  +& i[g(x)\sin{\left(  \varphi (x) t \right)} + ih(x)\cos{\left(  \varphi (x) t \right)}].
\end{split}
\end{equation}
For the first term in Eq.~\eqref{f(x)exp(iphi(x))}, since $g\in L^1([a,b])$, $\forall \,\epsilon > 0, \exists$ a \emph{staircase} function $g_1(x)$ satisfying
\begin{equation*}
  \int_{a} ^{b} |g(x)-g_1(x)| \, \dif{x} < \epsilon.
\end{equation*}
Hence
\begin{eqnarray*}
  \left|\int_{a} ^{b} g(x) \cos{(\varphi(x)t)} \, \dif{x} - \int_{a} ^{b} g_1(x) \cos{(\varphi(x)t)} \, \dif{x} \right| \\\leq \int_{a} ^{b} |g(x)-g_1(x)| \, \dif{x} < \epsilon.
\end{eqnarray*}
Because $g_1$ is a staircase function, according to Lemma~\ref{lemma3},
\begin{equation*}
    \lim\limits_{t \rightarrow \infty} \int_a ^b g_1(x)\cos{\left(  \varphi (x) t \right)} \, \dif{x} = 0.
\end{equation*}
Therefore,
\begin{equation*}
    \lim\limits_{t \rightarrow \infty} \int_a ^b g(x)\cos{\left(  \varphi (x) t \right)} \, \dif{x} = 0.
\end{equation*}
The conclusions are the same for the other three terms in Eq.~\eqref{f(x)exp(iphi(x))}.
\end{proof}

\begin{lemma}[Lebesgue's dominated convergence theorem]\label{RL D}
Suppose $\{f_n (x)\}\subset \mathfrak{M} ([a,b])$ and $\lim\limits_{n \rightarrow \infty} f_n(x)=f(x)$, a.e. $x\in[a,b]$. If $f(x)$ is dominated by an integrable function $g(x)$ in the sense that $|f(x)|\leq g(x)$, then
 \begin{equation*}
   \lim\limits_{n \rightarrow \infty} \int_a ^b f_n (x)  \, \dif{x} = \int_a ^b f (x)  \, \dif{x}.
 \end{equation*}
Further, the index can be continuous, namely,
 \begin{equation*}
   \lim\limits_{t \rightarrow \infty} \int_a ^b f (x,t)  \, \dif{x} = \int_a ^b f (x)  \, \dif{x}.
 \end{equation*}
\end{lemma}
Now we prove Lemma~\ref{th1}:
\begin{proof}
Set
\begin{equation*}
  I(x,t) =  \int_c ^d f(x,y) e^{i  \varphi (x,y) t } \, \dif{y}.
\end{equation*}
According to Lemma~\ref{lemma4}, one can find that
\begin{equation*}
  \lim\limits_{t \rightarrow \infty} I(x,t) = 0,\, a.e. x\in[a,b].
\end{equation*}
Besides,
\begin{equation*}
\begin{split}
  |I(x,t)| &=  \left|\int_c ^d f(x,y) e^{i  \varphi (x,y) t } \, \dif{y}\right|\\
  & \leq \int_c ^d |f(x,y)|\, \dif{y} = G(x),
\end{split}
\end{equation*}
where $G(x)$ is integrable. According to Lemma~\ref{RL D}, we get
\begin{equation*}
  \lim\limits_{t \rightarrow \infty} \int_a ^b \int_c ^d f(x,y) e^{i  \varphi (x,y) t } \, \dif{y}  \, \dif{x} = \lim\limits_{t \rightarrow \infty} \int_a ^b I(x,t)  \, \dif{x} = 0.
\end{equation*}
\end{proof}

\section{Proof for the sufficient condition}
Denoting the normalized eigenvectors of the matrix $C$ as $\ket{\varphi_m}$ and corresponding eigenvalues as $e^{i\omega_m}$ ($m=1,2,3$), we can construct a unitary matrix $P$ by arranging three eigenvectors $\ket{\varphi_m}$ together, namely, $P=[\ket{\varphi_1},\ket{\varphi_2},\ket{\varphi_3}]$. Then
\begin{eqnarray*}
P{\tilde U}P^\dagger&=&(PMP^\dagger)(PC^\dagger P^\dagger)(P M^\dagger P^\dagger) (P C P^\dagger)\\
&=&(P
\begin{bmatrix}
 1 & 0 & 0 \\
 0 & e^{-i \text{k}} & 0 \\
 0 & 0 & e^{-i \text{l}} \\
\end{bmatrix}
P^\dagger)
\begin{bmatrix}
 e^{-i\omega_1} & 0 & 0 \\
 0 & e^{-i\omega_2} & 0 \\
 0 & 0 & e^{-i\omega_3} \\
\end{bmatrix}\\
&&(P
\begin{bmatrix}
 1 & 0 & 0 \\
 0 & e^{i \text{k}} & 0 \\
 0 & 0 & e^{i \text{l}} \\
\end{bmatrix}
P^\dagger)
\begin{bmatrix}
 e^{i\omega_1} & 0 & 0 \\
 0 & e^{i\omega_2} & 0 \\
 0 & 0 & e^{i\omega_3} \\
\end{bmatrix}.
\end{eqnarray*}$\,$\par
We denote the unitary matrix
$$(P
\begin{bmatrix}
 1 & 0 & 0 \\
 0 & e^{-i \text{k}} & 0 \\
 0 & 0 & e^{-i \text{l}} \\
\end{bmatrix}
P^\dagger)
\begin{bmatrix}
 e^{-i\omega_1} & 0 & 0 \\
 0 & e^{-i\omega_2} & 0 \\
 0 & 0 & e^{-i\omega_3} \\
\end{bmatrix}$$
as $J$. Given that eigenvectors $\ket{\varphi_m}$ are all real vectors, the matrix $P$ would be a real matrix. Then, $P{\tilde U}P^\dagger=JJ^*$. The determinant of $JJ^*-\mathbb{I}$ is
\begin{equation*}
\text{det}(JJ^*-\mathbb{I})=
\text{det}(J(J^*-J^\dagger))=\text{det}(J)\text{det}(J^*-J^\dagger)
\end{equation*}
Since the matrix $J^*-J^\dagger$ is an anti-symmetrical three-dimensional matrix, its determinant is zero. Therefore, we have proved that 1 is one of the eigenvalues of the matrix $P{\tilde U}P^\dagger$, as well as the matrix $\tilde U$, because unitary transformations will not change the eigenvalues of an operator.\par

\end{document}